\documentclass[journal]{IEEEtran}
\usepackage[cmex10]{amsmath}
\usepackage{graphicx}
\usepackage{multirow}
\usepackage{epsfig,color}
\usepackage{url}
\usepackage{algorithmic}
\usepackage{algorithm}
\usepackage[numbers, sectionbib,comma, sort&compress,square]{natbib}
\usepackage{amssymb}
\usepackage{amsmath}
\usepackage{graphicx}
\usepackage{array}
\usepackage{tabularx}
\usepackage{textcomp}
\usepackage{color}

\usepackage{amsthm}

\usepackage{url}

\newtheorem{theorem}{Theorem}
\newtheorem{cor}{Corollary}
\newtheorem{lem}{Lemma}

\newtheorem{defn}{Definition}
\newtheorem{rem}{Remark}

\newcommand{\Al}{\alpha}

\newcommand{\pa}[1]{\left( #1 \right)}
\newcommand{\pac}[1]{\left\{ {#1} \right\}}

\newcommand{\beq}[2]{\begin{equation}\label{#1} #2 \end{equation}}
\newcommand{\bal}[2]{{\setlength\arraycolsep{2pt}\begin{eqnarray}\label{#1} #2 \end{eqnarray}}}
\newcommand{\beqn}[2]{\setlength\arraycolsep{2pt}\begin{equation*}\label{#1} {#2} \end{equation*}}
\newcommand{\baln}[2]{\setlength\arraycolsep{1.5pt}\begin{eqnarray*}\label{#1} #2 \end{eqnarray*}}

\newcommand{\iid}{i.i.d.}

\newcommand{\nn}{\nonumber}

\newcommand{\barr}[2]{\left\{\begin{array}{#1}#2 \end{array}\right.}

\newcommand{\bepsf}[2]{\begin{center}\epsfig{figure=#1,width=#2 cm}\end{center}}

\include{conferenceList_full}
\title{Design and Analysis of Stability-Guaranteed PUFs}
\author{\IEEEauthorblockN{Wei-Che Wang, Yair Yona, Suhas Diggavi and Puneet Gupta }
\IEEEauthorblockA{Department of Electrical Engineering, University of California, Los Angeles}
\{weichewang, suhasdiggavi, yairyo99, puneetg\}@ucla.edu
}

\begin{document}
\maketitle
\begin{abstract}
The lack of stability is one of the major limitations that constrains Physical Unclonable Function (PUF) from being put in widespread practical use. In this paper, we propose a weak PUF and a strong PUF that are both completely stable with 0\% intra-distance. These PUFs are called Locally Enhanced Defectivity Physical Unclonable Function (LEDPUF). The source of randomness of a LEDPUF is extracted from locally enhance defectivity without affecting other parts of the chip. A LEDPUF is a pure functional PUF that does not require helper data, fuzzy comparator, or any kinds of correction schemes as conventional parametric PUFs do. A weak LEDPUF is constructed by forming arrays of Directed Self Assembly (DSA) random connections is presented, and the strong LEDPUF is implemented by using the weak LEDPUF as the key of a keyed-hash message authentication code (HMAC). Our simulation and statistical results show that the entropy of the weak LEDPUF bits is close to ideal, and the inter-distances of both weak and strong LEDPUFs are about 50\%, which means that these LEDPUFs are not only stable but also unique.

We develop a new unified framework for evaluating the level of security of PUFs, based on password security, by using information theoretic tools of guesswork. The guesswork model allows to quantitatively compare, with a single unified metric,  PUFs with varying levels of stability, bias and available side information. In addition, it generalizes other measures to evaluate the security level such as min-entropy and mutual information. We evaluate guesswork-based security of some measured SRAM and Ring Oscillator PUFs as an example and compare them with LEDPUF to show that stability has a more severe impact on the PUF security than biased responses. Furthermore, we find the guesswork of three new problems: Guesswork under probability of attack failure, the guesswork of idealized version of a message authentication code, and the guesswork of strong PUFs that are used for authentication.

\end{abstract}
\let\thefootnote\relax\footnotetext{This work was supported in part by NSF grants 1136174 and 1321120.}


\vspace{-0.2in}

\section{Introduction}
A Physical Unclonable Function (PUF) is a small piece of circuitry such that its behavior, or Challenge Response Pair (CRP) \cite{Roel2010}, is uniquely defined and it is hard to be predicted and replicated because of the intrinsic random physical nature and the uncontrollability of process variations. As a security primitive, PUF can enable low overhead hardware identification, tracing, and authentication during the global manufacturing chain. The first PUF was introduced more than a decade ago \cite{Blaise02}. Since then, many silicon PUF implementations have been proposed, such as Arbiter PUF \cite{Suh07}, Ring Oscillator (RO) PUF \cite{Yin13}, SRAM PUF \cite{Holcomb09}, and many other variations. 

Since the key commonality between all current silicon PUF implementations is their use of {\em parametric} manufacturing variations, two of the most challenging design tasks a PUF designer will encounter are :
\begin{enumerate}
    \item How to make the PUF unique and stable even under extreme conditions without expensive implementation cost?
    \item How to evaluate the security level of a PUF given its uniqueness and stability measurements?
\end{enumerate}

\vspace{-0.2in}
\subsection{Limitations of Parametric PUFs}
\subsubsection{Random Local Variation Extraction}
One of the major concerns of parametric PUFs is that local variation should be the {\em only} entropy source for these PUFs \cite{Maiti2009}. However, from our experiments on a large silicon data set \cite{Lerong11}, only 13\% of total variation is random local variation, which means that most variation is coming from global or spatial variation. Any attempt to use global or spatial variation as the source of randomness can make them vulnerable to a class of {\em process side channel attacks}. For instance, two PUFs on the same (X,Y) location on different wafers are highly correlated (due to large wafer-level systematics present in most modern fabrication processes). As a result, a few sacrificial wafers can aid in developing a relatively straightforward side channel attack. We tested this side channel attack on silicon RO PUF measurements in 65nm technology across 300 wafers. Figure \ref{figure:attack} shows that the inter-distance \cite{Maiti2009} on the same (X,Y) is much smaller than the inter-distance across all PUFs. Therefore, an adversary with possession of a reference PUF, which is fabricated at the same (X,Y) location as the target PUF, would have a higher probability of guessing the correct answer than random guessing. The radial nature of systematic across wafer variation \cite{Lerong11} means that just a few reference PUFs drawn carefully may be sufficient for attackers instead of keeping full sacrificial wafers. 

\begin{figure}[htb]
\centering
\includegraphics[width=3.5in]{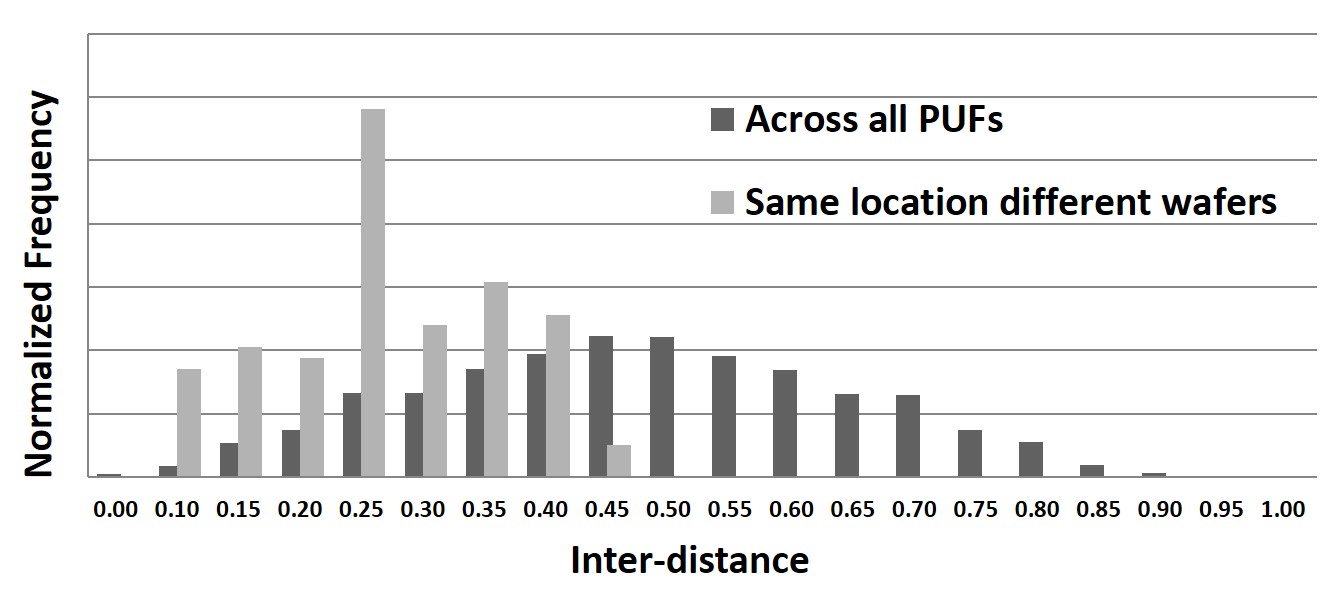}
\caption{The inter-distance of PUFs from same (X,Y) location on different wafers is much smaller than that of across all PUFs, which demonstrates a possible side channel attack.}
\label{figure:attack}
\end{figure}

\subsubsection{Measurement Noise}
Measurement noise could be another big issue for parametric PUFs and needs to be carefully compensated. For instance, metastability of the arbiter circuit for Arbiter PUFs and accumulated jitter in RO PUFs can be sources of measurement noises. For weak PUF measurement, we evaluate the intra-distance \cite{Maiti2009} of SRAM PUFs using fifteen commercial 45nm SOI test chips, where each consists 176kB data memory. The power-up state is measured 10 times during an 8-hour period, and the mean of intra-distance distribution is 2.57\%. Since the experiment is done in room temperature with exactly same settings, the difference is essentially contributed by the measurement noise.

\subsubsection{Environmental fluctuations and wearout}
Existing silicon PUFs are in nature susceptible to environmental fluctuations \cite{LangLin12} and wearout \cite{Nazhandali14}. To account for the instability issue, techniques such as error correction code (ECC), helper data or fuzzy comparator must be applied. A possible worst case scenario is when the environmental factors change significantly but yet remain constant. For instance, the PUF is enrolled at 20\textdegree C and is verified at 80\textdegree C. In this case, a fuzzy extraction process may not be able to recover the initial PUF response, even for multiple samples of the PUF. 

\vspace{-0.1in}
\subsection{Techniques to Improve Parametric PUF Quality}
A variety of techniques have been intensively studied over the years to extract random local variations or to make a PUF more stable and reliable. A Non-Volatile Memory (NVM) based PUF without helper data is presented in \cite{Wenjie14}. However, besides its hardware and calibration overhead, the results of uniqueness and entropy analysis are also missing. In \cite{Chi-En13}, the local randomness is distilled by modeling and subtracting the systematic variation. A similar technique is to subtract the averaged frequency from multiple measurements to reveal the true local random variation \cite{Feiten2015}. However, the calculation and information storage requirement come with the cost of addition latency and hardware. Taking the majority vote \cite{Majzoobi2010} or finding stable responses \cite{PotkonjakM14} are possible techniques to eliminate the measurement noise, however, at the cost of large latency or reduced number of challenges. Other complex implementations have been proposed to mitigate stability issues that often induce lower hardware efficiency \cite{Maiti2009}, additional circuit complexity \cite{Rahman2014}, or making the PUF more susceptible to attacks \cite{Jeroen2014}. Also, to protect PUFs from the worst case scenario as described creates a huge overhead as it requires to employ very strong ECC \cite{Xinmiao2015}.

\vspace{-0.1in}
\subsection{Quantifiable Security Evaluation Model}
In order to impersonate the hardware, the PUF attacker needs to
respond to a challenge with a correct response (\emph{i.e.,} guess a
secret).  Comprehensive security models for PUFs are described in
\cite{Armknecht2016}, including a precise identification of required
PUF properties, such as indistinguishably and
tamper-resilience. Though this specifies the security requirements, as
a ``checklist'', we believe that a more quantitative assessment of PUF
security can be valuable for both PUF designers and PUF users.  Inter-
and intra- fractional Hamming Distance (FHD) \cite{Maiti2009}, and
other statistical tests for randomness \cite{Andrew10}, have been used
for quantifying PUF security.  Though it is reasonable that having
larger inter-FHD is more secure, it does not tell the PUF designer
how much more secure it is. For
example, is it worth raising the inter-FHD from 40\% to 49\% at a cost
of extra hardware? In this work we present a more principled way to analyze PUF
security by connecting it to how one could evaluate password security,
through a guesswork framework. We derive a theoretical framework for
PUF security evaluation that brings together two important
properties of existing PUFs: predictability and
reproducibility \cite{PUFStudy2010}. This framework enables a
unified security quantification of several effects including bias,
noise, and side-channels on PUFs, as well as the security over
multiple challenge-response pairs, providing design guidance by
quantifying the security level of a PUF.

One can think of PUF signatures
like passwords, and its breakability should be evaluated by how strong it is, for example, how many attempts (on average) does it
  take to compromise it. This \emph{guessing} framework has been
  studied in the information theory literature and has recently been
  adopted by NIST as a measure of password security \cite{NISTPassword}. We bring this framework to evaluating the security of PUFs.
\vspace{-0.1in}
\subsection{The LEDPUF}
The issues of parametric PUFs, such as the described instability, wearout, measurement noise, limited local variation, and limited side channel attack resiliency, clearly motivate the need to design PUFs that do not rely on parametric performance variations as the entropy source. 


Rather than comparing parameter deviations, the response of an LEDPUF is stability-guaranteed because it depends on the local randomness enhanced through locally manipulated physical layout designs. We leverage intrinsic randomness in fabrication processes to generate hard defects, such as random permanent connections generated in Directed Self Assembly (DSA) process, which is highly compatible with CMOS technology and is expected to be used in manufacturing in the near future \cite{ITRS}.


Compared to similar parametric PUFs such as hardware obfuscation \cite{JamesB14} or digital PUFs \cite{TengXu14}, LEDPUF is completely stable and less susceptible to side channel attacks or model building attacks. The proposed LEDPUF is also a functional PUF where logic function itself is the signature and the strong LEDPUF can generate a variety of challenge-response pairs as needed. The Boolean nature of the response without any parametric dependence means that LEDPUF is not only immune to measurement noise and wearout, but also offers a greater level of reliability compared to existing PUFs as the output is resistant to changes in the environmental factors.

\subsection{The Contributions of this Paper}
The contributions of this paper are:
\begin{itemize}
\item The first stability-guaranteed silicon PUF through locally enhanced defectivity is proposed.
\item Detailed constructions of the weak and strong LEDPUF using random DSA connections are presented. LEDPUF is the first PUF with 0\% intra-distance without using any stability enhancement techniques.
\item We present a new unified framework for evaluating the level of security of PUFs through guesswork analysis. This framework enables us to evaluate and quantify the effect of noise, bias and model attacks on the level of security. We also relate guesswork to other security measures such as min-entropy, and mutual information. The model quantitatively measures the security of various PUFs under different scenarios, and by doing so enables us to compare the security level of different sorts of PUFs.
\item Security of noisy SRAM PUF and RO PUF implemented on real chips and Field Programmable Gate Arrays (FPGAs) are compared with a stable LEDPUF based on the new model.
\item Guesswork is derived for three new problems: Guesswork under probability of attack failure; the guesswork of strong PUFs under model building attacks; and the guesswork of an idealized version of a message authentication code.
\end{itemize}

The rest of the paper is organized as follows: In Section \ref{DSA_ran}, we describe the local randomness extraction from DSA. In Section \ref{LEDPUF_cons}, the structure of the stable signal unit (SSU) is first described, followed by the construction of the proposed weak LEDPUF from the SSU. The unified guesswork framework is presented in Section \ref{sec:GuessworkFramwork}. In Section \ref{exp}, the evaluation results of weak LEDPUF and nosiy weak PUFs are presented along with a derivation of the guesswork under  probability of attack failure. The construction of strong LEDPUF, evaluation of its security, as well as a derivation of the guesswork of message authentication codes is presented in Section \ref{sec:GWStrongLEDPUF}. Finally, we conclude the paper in Section \ref{conclusion}.



\vspace{-0.1in}
\section{DSA Randomness Extraction} \label{DSA_ran}
We propose DSA as one of the sources of stable randomness for LEDPUF. DSA is a mechanism by which immiscible block copolymers (BCP) phase-separate into certain structures \cite{xu2009self}. The guiding templates, which are used to {\em guide} the self-assembly process \cite{jarnagin_high_2013}, can be lithographically-printed trenches (Graphoepitaxy) or chemically-treated surfaces (Chemoepitaxy). Once a guiding template is printed, its surface is spin-coated with the BCP solution. The phase separation occurs during the thermal annealing where thermal equilibrium is achieved when the free energy is minimized. With a particular BCP and surface treatment of substrate \cite{kim2013interplay}, cylinders are formed of one block in a matrix of the other block \cite{yasmine2015}.

Since the minimum energy state strongly depends on the level of confinement achieved by the layout of guiding templates, for bigger-sized templates, it becomes energetically less expensive to induce a defect than to achieve a defect-free energy minimization \cite{black_polymer_2007, kang2011degree, sundrani2004hierarchical}. Therefore, final assembly results can be random by designing guiding templates that are large enough to cause random assembly errors even if there are no lithographic variations \cite{Wang16}. Figure \ref{fig:DSA1} shows three simulation results of the same large guiding template with an existing DSA simulator \cite{Brandon2014}, where the model of the PS-b-PMMA copolymer has been validated in \cite{Detcheverry08}. The three layers inside the polygonal guiding template are the top, middle, and bottom layers of a via. If a cylindrical via hole is formed correctly, the three layers should be three overlapped concentric circles. However, for the large guiding template, random arrangement with different orientations begin to occur. In other words, the randomness of DSA is confined within predetermined local areas only by deliberately designing "bad" guiding templates.

\begin{figure}[h]
\begin{centering}
\includegraphics[width=3.0in]{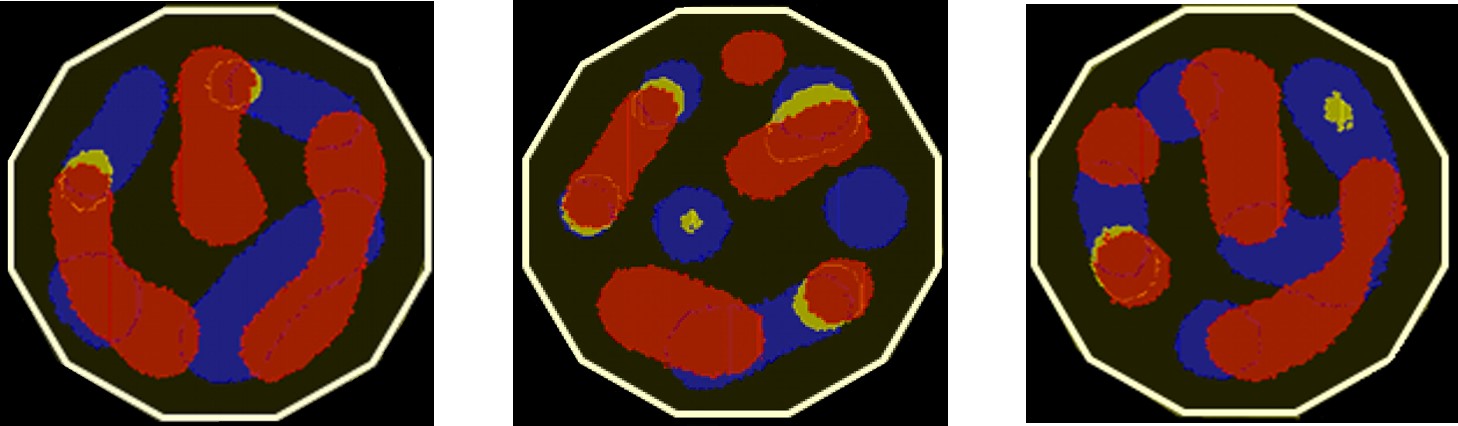}
\par\end{centering}
\caption{Random via formations with a same large guiding template. }
\label{fig:DSA1}
\end{figure}

\vspace{-0.1in}
\subsection{Hard Defective Connection Formation} \label{d_via}
We leverage the randomness extracted from DSA to form randomly assembled connections, and these connections are then used to fabricate LEDPUF. Though in conventional DSA, the goal of the guiding template design is to achieve high confinement and avoid regions of random phase transitions, we use the same theory but to enhance randomness in assembly. To construct a DSA random hard defective connection, we configure the size of the guiding template so that two vias are formed with a certain probability that they are connected permanently. A DSA hard defective connection is composed of the two vias along with the connection. 

In our experiment, each simulation contains three guiding templates with a same shape, and two vias are formed in each of the guiding template. If the via pair in a same guiding template is merged, the DSA hard defective connection is in closed state; otherwise, the connection is in opened state. In our statistical analysis, an open state is represented by a logic ``1'', and a closed state is represented by a logic ``0''. 500 simulation were performed in our experiments, so 1500 bits of raw data is obtained from the simulation. Based on our simulated data, the empirical entropy of triples of bits is only 0.0063 bits smaller than the entropy of independent triple of bits, which implies that the states of the three connections in a simulation is independent with each other as expected in real DSA process \cite{Detcheverry08}. Examples of a simulation result in 2D and 3D views are shown in Figure \ref{fig:D1} (a) and (b), respectively. In the 2D view, the rectangular shapes with rounding corners are the guiding templates, and the shapes inside of the guiding templates are the vias. If the via pair in a same guiding template is merged, the DSA hard defective connection is formed as shown in Figure \ref{fig:D1} (c), and it is in permanent closed state; otherwise, the DSA hard defective connection is in permanent opened state as shown in Figure \ref{fig:D1} (d). In Figure \ref{fig:D1} (a) and (b), two hard defective connections are in opened state, and one connection is in closed state.

\begin{figure}[h]
\begin{centering}
\includegraphics[width=3.0in]{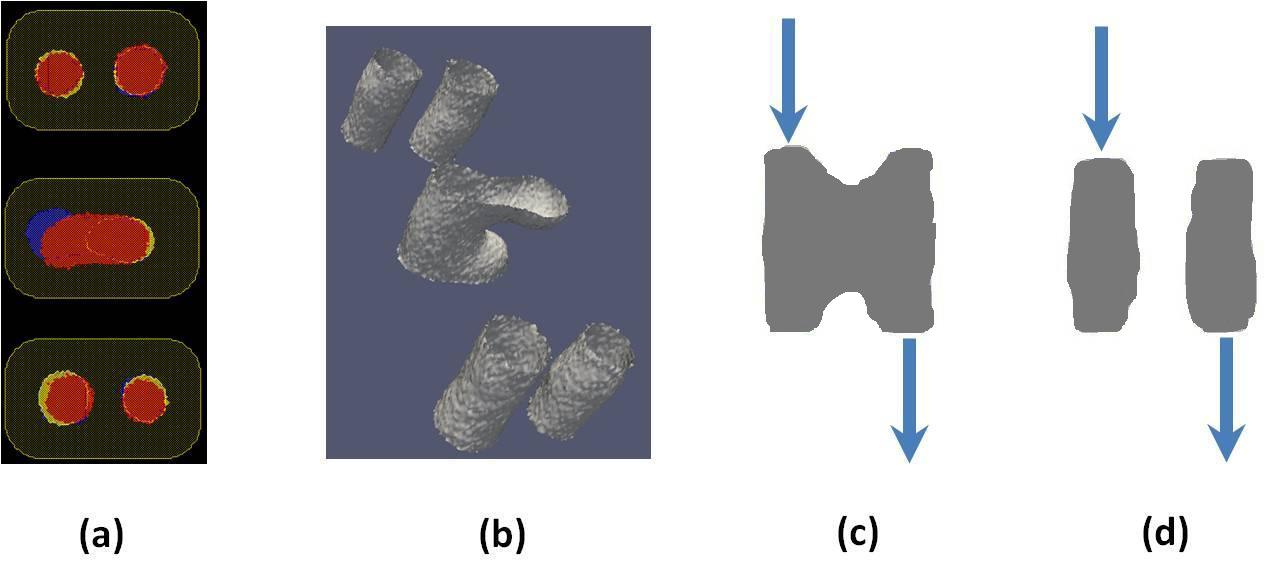}
\par\end{centering}
\caption{(a) 2D view of 3 DSA hard defective connections. (b) 3D view of 3 DSA hard defective connections. (c) Vias are partially merged, so the connection is in permanent closed state. (d) Vias are not merged, so the DSA hard defective connection is in permanent opened state.}
 \label{fig:D1}
\end{figure}

\vspace{-0.1in}
\section{Weak LEDPUF Construction} \label{LEDPUF_cons}
The proposed weak LEDPUF is composed of arrays of SSUs. Each SSU is constructed from a DSA defective connection, which can be considered as random switches with permanent states that determine the unique and stable function of the circuit. Figure \ref{fig:r3} (a) shows the implementation of a SSU. Two ends of the DSA connection are connected to VDD and GND through opposite switches. Figure \ref{fig:r3} (b) shows the abstraction of a SSU. In standby mode or before the evaluation, the evaluation signal EVA is low and the output is zero. During evaluation mode, EVA becomes high, and the output is either one or zero depending on the permanent state of the DSA connection. If the DSA connection is closed, the output is one; otherwise, the output is zero.

\begin{figure}[h]
\begin{centering}
\includegraphics[width=3.2in]{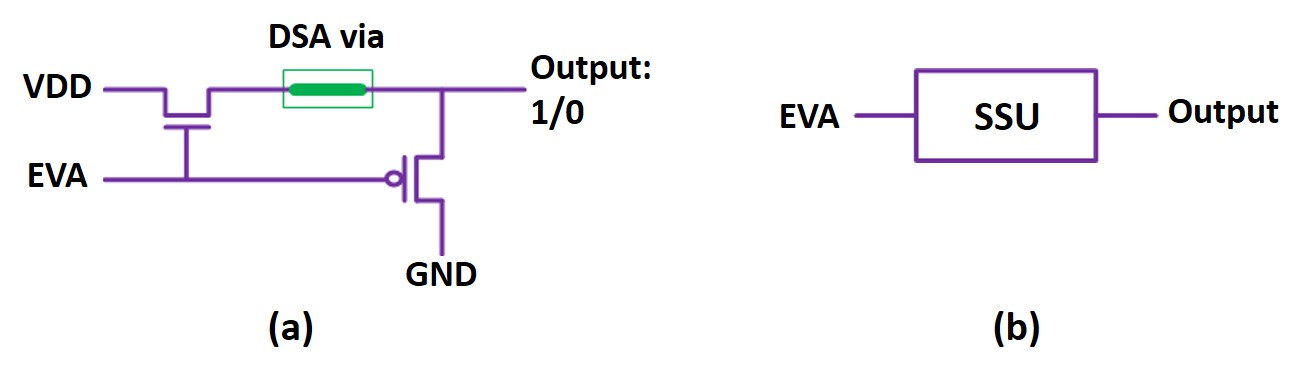}
\par\end{centering}
\caption{(a) Stable signal unit implementation. When EVA is high, the output is either one or zero permanently depending on the state of the DSA via. (b) Abstraction of a SSU.}
\label{fig:r3}
\end{figure}

The proposed weak LEDPUF is constructed by arranging the SSUs in forms of arrays. Figure \ref{fig:ssu_array} illustrates a weak LEDPUF with $n$ rows and $m$ columns, where the number of SSUs is $nm$, and the number of CRPs is $n$. Since only one of the rows is being evaluated at a time, a one-hot decoder is used so that only one bit of the EVA vector is logic 1. The challenge fed into the decoder is a $log(n)$-bit input, and the response is a $m$-bit output. 

\begin{figure}[h]
\begin{centering}
\includegraphics[width=3.3in]{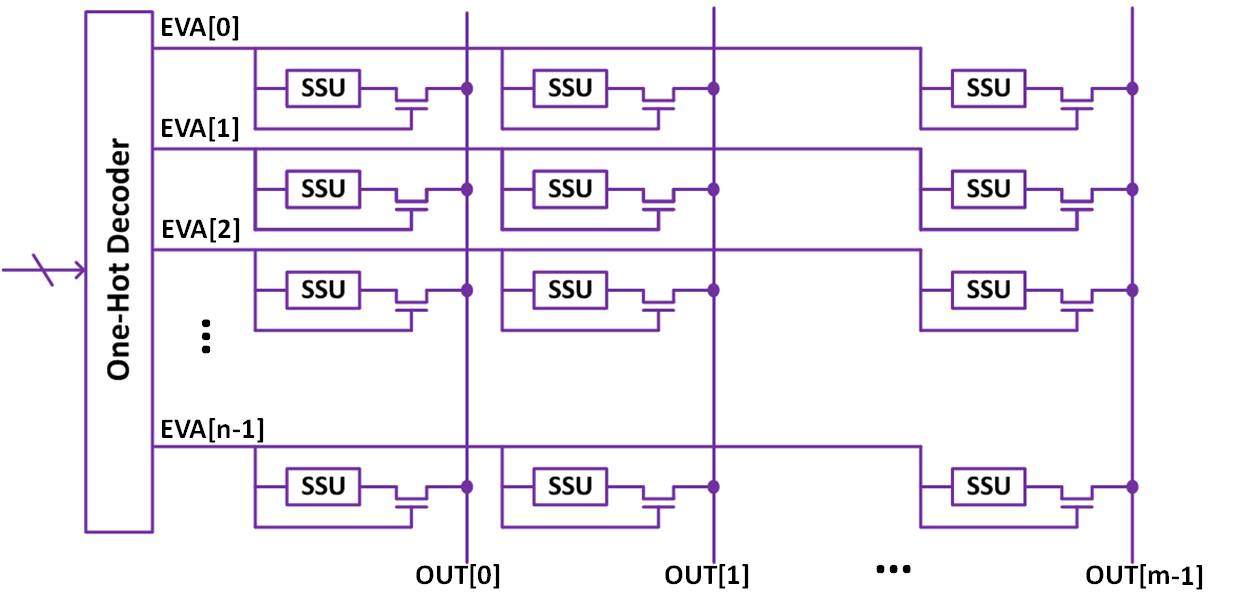}
\par\end{centering}
\caption{A weak LEDPUF with $n$ challenges and $m$-bit response. Only one bit of the EVA vector is logic 1 at a time.}
\label{fig:ssu_array}
\end{figure}

Compared with existing weak SRAM PUFs, the weak LEDPUF has several evident advantages:
\begin{itemize}
\item It is completely stable, so it has no area or latency overhead. To generate a bit response, the weak LEDPUF requires only one SSU and a transistor, or 3  transistors equivalently, as for a standard SRAM cell, 6 transistors are required. Once the state of the DSA via is determined, the output is fixed permanently, so no additional ECC, fuzzy extraction, or helper data is needed. As stated in \cite{Guajardo07}, for a SRAM PUF to generate a 128-bit response, more than 4k SRAM cells are needed under a condition with 15\% bit error probability. Therefore, the total number of transistors needed for SRAM PUF to generate a 128-bit response would be 24k, where for the weak LEDPUF, only 384 transistors are needed, which is more than 600x less than a SRAM PUF, thus the area is also much smaller even assuming that the hardware cost of the peripheral circuits are similar.

\item In addition to model building attacks \cite{Herder2014}, the weak LEDPUF is also more resistant to existing attacks to SRAM PUFs, such as laser stimulation \cite{Nedospasov2013} or Photonic Emission Analysis (PEA) \cite{Helfmeier2013}. The laser stimulation attack focuses on retrieving the on/off state of transistors, but for weak LEDPUF, the states of the transistors, which depend on the EVA signal, do not reveal secret information. The PEA attack does not work effectively because for each SSU, the source voltage (VDD) of the NMOS is always higher than the drain voltage, and the PMOS at the output will not stay in saturation region since the output will be pulled down even if the DSA connection is formed.
\end{itemize}

\vspace{-0.1in}
   
\section{Guesswork as A Unified Framework for Evaluating The Level of Security}\label{sec:GuessworkFramwork}
\subsection{Why Consider Guesswork?}\label{subsec:WhyConsiderGuesswork}
Consider the following game: Bob draws a sample $x$ from a random
variable $X$, and an attacker Alice who does not know $x$ but knows
the probability mass function $P_{X}\pa{\cdot}$, tries to guess it. An
oracle tells Alice whether her guess is right or wrong. This is the
situation where an attacker tries different passwords to access an account.

If Alice has only one guess, then the optimal strategy that maximizes
her chance of guessing $x$ successfully is choosing the most probable
$x$. In this case the chance of guessing $x$ is $\max_{x^{'}\in
  X}P_{X}\pa{x^{'}}$ and the \emph{predictability} of $X$ is given by
its min-entropy \cite{FuzzyExt2007}:
\beq{} {
  H_{\infty}\pa{P_{X}}=-\log_{2}\pa{\max_{x^{'}\in X}P_{X}\pa{x^{'}}}.  }

If Alice is allowed to make as many guesses as required until she
finds $x$, then the optimal strategy is guessing elements in $X$ based
on their probabilities in ascending order \cite{Arikan_Ineq_Guessing}.  It has
been shown that the average number of guesses it takes Alice to find
$x$ (denoted by $G\pa{X}$ and termed guesswork) is not given by the
traditional Shannon entropy \cite{Arikan_Ineq_Guessing}.  For example, when
drawing a random vector of length $m$, $\underline{X}$, which is
independent and identically distributed (i.i.d.) with distribution
$P=[p_1,\ldots,p_n]$, the exponential growth rate of the guesswork
scales according to the Renyi entropy $H_{\alpha}\pa{X}$ with parameter $\alpha=1/2$
\cite{Arikan_Ineq_Guessing}:  
\beq{eq:GuessworkExpGrowthRate} 
{
\lim_{m\to\infty}\frac{1}{m}\log_{2}\pa{E\pa{G\pa{X}}}=
 H_{1/2}\pa{P}=2\cdot\log_{2}\pa{\sum_{i}p_{i}^{1/2}} } where
$H_{\frac{1}{2}}\pa{P}\ge H\pa{P}$ with equality only for the uniform
probability mass function. 


The security of a PUF is predicated by the inherent random signature
in the hardware. An attacker wants to either impersonate a hardware by
guessing its random signature, or to learn it by eavesdropping. In order to impersonate the hardware, the attacker
needs to respond to a challenge with a correct response. In order to
evaluate the security of a PUF we connect to the framework for password security \cite{NISTPassword}. For a dictionary attack, a guessing framework quantifies security through the number of guesses the impersonator has to make in order to
identify the password (or inherent randomness) and therefore respond correctly to
all possible challenges. Therefore, we quantify the level of security of a PUF through the number of guesses required to break it.

In subsection \ref{subsec:GuessworkUnifiedFramework} we show that guesswork can serve as a unified framework for evaluating and quantifying the security of PUFs. Essentially, other measures of evaluating the level of security such as min-entropy and mutual information are special cases of guesswork. Furthermore, guesswork allows to quantify the level of security under more elaborated scenarios such as the security level when key stretching mechanism is used \cite{WagnerKeyStreatching} as well as when allowing an attack failure probability (this problem is presented in Subsection \ref{subsec:GuessworkUnifiedFramework}). 

Interestingly, characterizing the security of a PUF through guesswork reveals a new interplay between the bias of a PUF response, and the noise (due to instability) which is incorporated in each sample. Guesswork is very sensitive to the presence of instability, but yet is not very susceptible to bias. These properties are discussed in subsection \ref{subsec:MinEntVsGuesswork}. Therefore, guesswork highlights the advantage of stable PUFs over unstable PUFs, when evaluated through the number of guesses required to break a PUF. To the best of our knowledge this advantage has not been reported in literature. 


Moreover, we present a formal evaluation methodology for PUFs security, while identifying the impact of bias, noise and side-channels. 

\vspace{-0.1in}
\subsection{Background}\label{subsec:Background}
The guesswork $G\pa{X}$ is a random variable that represents the number of guesses required to guess a random variable $x$. Therefore, the probability of having $G\pa{x}$ guesses is $P_{X}\pa{x}$. The $\rho$th moment of guesswork is 
\beq{eq:BasicGuesswork}{
E\pa{G\pa{X}^{\rho}}=\sum_{x}G\pa{x}^{\rho}\cdot p_{X}\pa{x}.
}

The definition of guesswork can be extended to the case where the attacker has a side information $Y$ available. In this case the average guesswork for $Y=y$ is defined as $G\pa{X|Y=y}$, and the $\rho$th moment of $G\pa{X|Y}$ is
\beq{eq:ConditionalGuesswork}{
E\pa{G\pa{X|Y}^{\rho}}=\sum_{y}E\pa{G\pa{X|Y=y}^{\rho}}\cdot p_{Y}\pa{y}.
}

Massey \cite{Messy1994} noted that a dictionary attack minimizes the expected number of guesses (i.e., guessing the values in the decreasing order of $P_{X}\pa{x}$). Arikan \cite{Arikan_Ineq_Guessing} has bounded the $\rho$th moment of the optimal guesswork, $G^{\ast}\pa{X|Y}$, by 
\bal{eq:GeneralExp}{
\pa{1+\ln\pa{M}}^{-\rho}\sum_{y}\pa{\sum_{x}P_{X,Y}\pa{x,y}^{\frac{1}{1+\rho}}}^{1+\rho}
\le \nn\\
E\pa{G^{\ast}\pa{X|Y}^{\rho}}\le \sum_{y}\pa{\sum_{x}P_{X,Y}\pa{x,y}^{\frac{1}{1+\rho}}}^{1+\rho}
}
where $M=|X|$ is the cardinality of $X$. Furthermore, in \cite{Arikan_Ineq_Guessing} it has been shown that when $X$ and $Y$ are \iid, the exponential growth rate of the optimal guesswork is
\beq{eq:asymptiticconditionalguess}{
\lim_{m\to\infty}\frac{1}{m}\log_{2}\pa{E\pa{G^{\ast}\pa{X|Y}^{\rho}}}=\rho\cdot H_{\frac{1}{1+\rho}}\pa{P_{X,Y}\pa{x,y}}
}
where $m$ is the size of $X$ and $Y$, and  
\beq{}{
H_{\frac{1}{1+\rho}}\pa{P_{X,Y}\pa{x,y}}=\frac{1}{\rho}\log_{2}\pa{\sum_{y}\pa{\sum_{x}P\pa{x,y}^{\frac{1}{1+\rho}}}^{1+\rho}}} is Renyi's conditional entropy of order $\frac{1}{1+\rho}$ \cite{Arikan_Ineq_Guessing}.

Another extension of guesswork \cite{MerhavArikanDist} considers a game in which it is sufficient to guess $x$ up to a certain level of distortion $D$, according to some distance metric $d\pa{x,\hat{x}\pa{i}}=\sum_{i=1}^{m}d\pa{x_{i},\hat{x}\pa{i}}$. Essentially, when $G\pa{x}=i$, the word $\hat{x}\pa{i}$ is guessed such that $d\pa{x,\hat{x}\pa{i}}\le m\cdot D$, that is, when the attacker guesses a word which is within a Hamming distance $m\cdot D$ of $x$ the game is over. The authors in \cite{MerhavArikanDist} have solved this problem for the general case; more specifically, for a binary source which is drawn i.i.d.  Bernoulli$\pa{p}$ and Hamming distortion
\beq{}{
d\pa{x_{j},\hat{x}_{j}}=\barr{cc}{
1 & x_{j}=\hat{x}_{j}\\
0 & x_{j}\neq\hat{x}_{j}
}
}
where $1\le j\le m$, they have shown that the exponential growth rate of the guesswork equals
\bal{eq:BoundsAsymptoticCondGesswork}{
\lim_{m\to\infty}\frac{1}{m}\log_{2}\pa{E\pa{G^{\ast}\pa{X,D}^{\rho}}}=\rho\cdot E\pa{D,p}=\nn\\ \max\pa{\rho H_{\frac{1}{1+\rho}}\pa{p}-\rho\cdot H\pa{D},0}
}
where $H\pa{D}=-D\log_{2}\pa{D}-\pa{1-D}\log_{2}\pa{1-D}$ is the binary Shannon entropy \cite{CoverBook}.


Guesswork has been analyzed in many other scenarios such as guessing under source uncertainty with and without side information \cite{Sunderesan07}, \cite{SunderesanGuessSourceUncertaintySideInfo}, using guesswork to lower bound the complexity of sequential decoding \cite{Arikan_Ineq_Guessing}, guesswork for Markov chains \cite{MaloneGuessworkandEntropy}, and guesswork for multi-user systems \cite{Duffy15}.

\vspace{-0.2in}
\subsection{Extending Guesswork }\label{subsec:GuessworkUnifiedFramework}
In this subsection we extend the definition of guesswork and show that it can serve as a unified framework for evaluating the level of security of PUFs by incorporating noise and bias. In addition, we relate guesswork to other measures  such as mutual information and min-entropy.

We begin by finding the moments of the guesswork of a noisy weak PUF.
\begin{theorem}\label{th:TheoryForGuessworNoisyPUF}
When the response of a weak PUF is noisy such that the noise is additive and drawn i.i.d. Bernoulli$\pa{D}$, and the original response is i.i.d. Bernoulli$\pa{p}$ , the $\rho$th moment of the guesswork increases at rate $\rho\cdot E\pa{D,p}$ as defined in \eqref{eq:BoundsAsymptoticCondGesswork}.
\end{theorem}
\begin{proof}
The idea behind the proof is that guessing within Hamming distance $m\cdot D$ of the original response, enables the attacker to find the original response by using the helper-data.

Essentially there are two options. The first possibility is to construct a code to guess the original response up to Hamming distance $m\cdot D$ as is done in \cite{MerhavArikanDist}, and then use the helper-data in order to find the original response, in which case the rate of the $\rho$th moment is $\rho\cdot E\pa{D,p}$. The second possibility is to use the helper data (e.g., the coset of an ECC) to guess over a subset. In this case, since the helper-data  breaks up the space into subspaces of size $2^{\pa{1-H\pa{D}}\cdot m}$ \cite{CoverBook}, guessing through the subspace can only bring the rate of the $\rho$th moment down to $\rho\cdot E\pa{D,p}$. Therefore, the minimal rate is $\rho\cdot E\pa{D,\rho}$.  
\end{proof}

Before we present a new game that extends the traditional definition of guesswork, let us define the type of a vector.
\begin{defn}
Consider a binary vector $x$ of size $m$ and assume that $N\pa{x|1}$ is the number of elements of this vector that are equal to 1. In this case when $N\pa{x|1}/m=q$ the vector $x$ is of type $q$.  
\end{defn}

We now define a new guessing game that captures different measures for evaluating the level of security.
\begin{defn}[Guesswork under attack failure constraint]
Consider the following game: Bob draws a vector $x$ of size $m$ i.i.d. Bernoulli$\pa{p}$. The attacker Alice 
has to guess $x$ up to Hamming distance $m\cdot D$ as defined in subsection \ref{subsec:Background}, under the constraint that the probability of attack failure is smaller than or equal to $2^{-\Al\cdot m}$ where $\Al\ge 0$, that is, Alice may decrease the number of guesses by guessing only a subset of all possible words, which leads in turn to a certain probability of attack failure. We define the optimal guesswork for this game as $G^{\ast}\pa{X;D,\Al}$. Furthermore, we define the guesswork in the case where the probability of attack failure is zero as $G^{\ast}\pa{X;D,\infty}=G^{\ast}\pa{X;D}$.  
\end{defn}

\begin{rem}
The relation between $G^{\ast}\pa{X;D,\Al}$ and previous works is as follows: 
\begin{itemize}
\item $Pr\pa{G^{\ast}\pa{X;0,\infty}=1}=2^{-m\cdot H_{\infty}\pa{p}}$, that is, the min-entropy. 
\item $\lim_{m\to\infty}\frac{1}{m}\log\pa{E\pa{G^{\ast}\pa{X;D,\infty}^{\rho}}}=\rho\cdot E\pa{D,p}$ as defined in \eqref{eq:BoundsAsymptoticCondGesswork}.
\end{itemize}
\end{rem}

The following theorem characterizes a lower bound for $G^{\ast}\pa{X;D,\Al}$ for any moment $\rho >0$ in the case where the attacker is allowed not to guess certain types.
\begin{theorem}\label{th:GWWithErrors}
\baln{}{
&\lim_{m\to\infty}&\frac{1}{m}\log\pa{E\pa{G^{\ast}\pa{X;D,\Al=D\pa{s||p}}^{\rho}}}\nn\\
&\le&\barr{cc}{\rho\cdot\pa{ H_{\frac{1}{1+\rho}}\pa{p}- H\pa{D}} &s^{\ast}\le s\le 1\\
    \rho\cdot\pa{ H\pa{s}-H\pa{D}}-D\pa{s||p} &p<s\le s^{\ast}
}
}
$0\le D\le p\le 1/2$, and 
\baln{}{
&\lim_{m\to\infty}&\frac{1}{m}\log\pa{E\pa{G^{\ast}\pa{X;D,\infty}^{\rho}}}\nn\\
&=&\rho\cdot\pa{H_{\frac{1}{1+\rho}}\pa{p}-H\pa{D}}
}
where $s^{\ast}=\frac{p^{\pa{1+\rho}^{-1}}}{p^{\pa{1+\rho}^{-1}}+\pa{1-p}^{\pa{1+\rho}^{-1}}}$, the probability of attack failure decreases like $2^{-m\cdot \Al}$, \beqn{}{
D\pa{s||p}=s\cdot\log_{2}\pa{s/p}+\pa{1-s}\cdot\log_{2}\pa{\pa{1-s}/\pa{1-p}}
}
is the Kullback-Leibler divergence \cite{CoverBook}, and Alice chooses a set $A=\pac{q_{1},\dots,q_{L}}$ of types whose vectors are not guessed, such that the probability that $N\pa{x|1}/m\in A$ is smaller than or equal to $2^{-\Al\cdot m}$, that is, Alice guesses words in $A^{C}$.
\end{theorem}
\begin{proof}
The proof is in Appendix \ref{sec:ProofThGWWithErrors}
\end{proof}
Note that $\rho\cdot H\pa{s^{\ast}}-D\pa{s^{\ast}||p}=\rho\cdot H_{\frac{1}{1+\rho}}\pa{p}$

The next three remarks point out a few properties of $G^{\ast}\pa{X;D,\Al}$.
\begin{rem}
When an attacker attempts to break a very large number of independent PUF responses (or passwords), where the probability of attack failure is $2^{-m\cdot\Al}$, he is very likely to break a fraction of $1-2^{-m\cdot\Al}$ of the PUF responses (passwords), and this in turn leads to average guesswork across PUF responses (passwords) that increases at a rate 
$\lim_{m\to\infty}\frac{1}{m}\log\pa{E\pa{G^{\ast}\pa{X;D,\infty}^{\rho}}}\le \rho\cdot\pa{H_{\frac{1}{1+\rho}}\pa{p}-H\pa{D}}=\rho\cdot E\pa{D,p}$, when $D\le p$.
\end{rem}
\begin{rem}
When $s=p+\epsilon$ the average guesswork is approximately $H\pa{p}-H\pa{D}$, which is the rate distortion function of Hamming distortion \cite{CoverBook}.
\end{rem}
\begin{rem}
Note that when $p=1/2$ also $s^{\ast}=1/2$ and the upper bound of Theorem \ref{th:GWWithErrors} is equal to $\lim_{m\to\infty}\frac{1}{m}\log\pa{E\pa{G^{\ast}\pa{X;D,\infty}^{\rho}}}$, that is, when guessing according to the method presented in Theorem \ref{th:GWWithErrors} Alice does not gain anything from having a failure probability larger than zero.
\end{rem}

We now derive an expression for the min-entropy when the attacker has to guess a word that is within Hamming distance $m\cdot D$ of the password. In this case the min-entropy of a binary i.i.d. source subject to Hamming distortion $D$ is equivalent to choosing a word which is in a ball of radius $m\cdot D$ around the most likely word (i.e., the probability of guessing a word which is in the most likely ball). The asymptotic value of the min-entropy subject to distortion $D$ is given by the following lemma.

\begin{lem}\label{lem:ProofofMinEntropysubjecttoDistortion}
Consider a binary word of length $m$ for which each element is drawn i.i.d. from Bernoulli$\pa{p}$. The min-entropy subject to Hamming distortion $D$ converges to   
\beq{}{
-\lim_{m\to\infty}\frac{1}{m}\log_{2}\pa{P_{X}^{ball}}=\barr{cc}{D\pa{D||p} & 0\le D\le p\\
0 & p<D\le 1} 
}
where $p\le 1/2$, $P_{X}^{ball}=\sum_{i=0}^{m\cdot D}\binom{m}{i} p^{i}\pa{1-p}^{m-i}$.
\end{lem}
\begin{proof}
The proof is in Appendix \ref{app:ProofofMinEntropysubjecttoDistortion}.
\end{proof}
\begin{rem}
Note that the result of Lemma \ref{lem:ProofofMinEntropysubjecttoDistortion} can also be interprested as $Pr\pa{G^{\ast}\pa{X;0,\infty}\le 2^{m\cdot H\pa{D}}}$ when a password of length $m$ is drawn i.i.d. Bernoulli$\pa{p}$, where $0 \le D\le 1/2$.
\end{rem}

\subsection{Examples for Quantifying the Security of PUFs}
In this subsection we present a few examples that illustrate how to use guesswork in order to quantify the level of security of PUFs. We address evaluations for unstable PUFs as well as for stable PUFs. We incorporate into the expressions noise, bias, and side information coming from other PUFs or from side channel/model attacks.

The first step in calculating the guesswork of a PUF is evaluating the probability function according to which it is drawn, as well as the noise level. In this subsection we assume that the bits are i.i.d. for which case the first step is evaluating the bias of the stable bits and then estimating the noise level of the unstable bits; evaluating the bias of the stable bits enables us to state that the PUF response is drawn i.i.d. from the probability function
\beq{}{
P_{0}=p\quad P_{1}=1-p
}
whereas the probability of transition of a bit when re-sampling a PUF is $q$, such that
\beq{eq:NoisyPUFExample}{
x^{\pa{2}}=x^{\pa{1}}\oplus e
}
where $x^{\pa{1}}$, $x^{\pa{2}}$ are the first and second samples of the unstable (noisy) PUF, and $Pr\pa{e_{j}=1}=q$, $1\le j\le m$.

For stable PUFs such as the one presented in this paper, it is sufficient to calculate the bias and assign the probability function to $\rho\cdot H_{\frac{1}{1+\rho}}\pa{P}$ in equation \eqref{eq:asymptiticconditionalguess}, in order to get the $\rho$th moment of guesswork. For example, when the stable PUF is drawn i.i.d. according to Bernoulli$\pa{0.47}$ the average guesswork of a PUF of large enough size ($m=256$, say)  is proportional to
\beq{}{
2^{H_{1/2}\pa{0.47}\cdot m}=2^{0.9987\cdot m }
}
whereas the largest guesswork that we can expect for is achieved by an unbiased PUF for which each bit is drawn i.i.d. Bernoulli$\pa{0.5}$, and is proportional to
\beq{}{
2^{H_{1/2}\pa{0.5}\cdot m}=2^{m}.
}

For unstable PUFs, re-sampling the PUF yields a noisy version of the original response as presented in equation \eqref{eq:NoisyPUFExample}. When the probability of transition is $q$, Theorem \ref{th:TheoryForGuessworNoisyPUF} shows us that it is sufficient to guess the original response $x^{\pa{1}}$ up to Hamming distance $m\cdot q$. The intra distance can be used to evaluate the noise level. For example, when considering an unbiased unstable PUF with a transition probability $q=0.1$, we get that the guesswork is proportional to 
\beq{}{
2^{m\cdot\pa{1-H\pa{0.1}}}=2^{0.531\cdot m}
}
which means that noise decreases the average number of guesses significantly.

The conditional guesswork \eqref{eq:ConditionalGuesswork} enables us to quantify the effect of side information on the security level of both stable and unstable PUFs. In order to evaluate the conditional guesswork we first need to characterize the conditional probability. The conditional probability depends on the type of attack which is being carried; in some cases characterizing its effect on the randomness of the response requires some effort. A simple example for a side information attack is one in which an attacker has another PUF which is correlated with the original one. For example, consider an unbiased stable PUF $x$ for which each element is drawn i.i.d. Bernoulli$\pa{0.5}$, and assume that an attacker has another unbiased stable PUF, $y$, which is correlated with $x$ such that 
\beq{}{
P\pa{y|x}=P\pa{e}
}
where $e$ is drawn i.i.d. Bernoulli$\pa{0.2}$. 
In this case the unconditional guesswork $G\pa{X}$ is proportional to 
\beq{}{
2^{H_{1/2}\pa{1/2}\cdot m}=2^{m}
}
whereas the conditional guesswork $G\pa{X|Y}$ is proportional to
\beq{}{
2^{H_{1/2}\pa{0.2}m}=2^{0.848\cdot m}
}
because of the fact that in this case $x$ given $y$ is also drawn i.i.d. Bernoulli$\pa{0.2}$. In general, the correlation between PUFs can be evaluated through the inter distance.

Conditional guesswork subject to distortion enables to evaluate the guesswork of an unstable PUF when side information is available. The method of evaluating the guesswork is similar to the previously mentioned methods for evaluating conditional guesswork and guesswork subject to distortion.

\subsection{The Effect of Noise Vs. The Effect of a Bias}\label{subsec:MinEntVsGuesswork}
In this subsection we analyze the expressions for guesswork as well as min-entropy subject to distortion, and quantify the impact that noise and bias have on PUFs. Furthermore, we show that the effect of noise is far worse than the effect of bias in terms of average guesswork.

First, let us focus on the effect of noise and bias on the expected value of the guesswork (i.e., the case where $\rho=1$). From Theorem \ref{th:TheoryForGuessworNoisyPUF} we get that when the transition probability is $D$, the asymptotic growth rate of the expected value of the guesswork is 
\beq{eq:UpperBoundGuessWithDist}{
H_{1/2}\pa{P_{X}}-H\pa{D}.
}
On the other hand the asymptotic growth rate of the expected value of the guesswork of a stable PUF whose bits are drawn i.i.d. from Bernoulli$\pa{p}$ is
\beq{eq:RenyiEntAsAMeaasureforGuess}{
H_{1/2}\pa{p}=2\cdot \log_{2}\pa{\sqrt{p}+\sqrt{1-p}}.
}
The first derivative of equation \eqref{eq:UpperBoundGuessWithDist} equals
\beq{}{
\log_{2}\pa{D}-\log_{2}\pa{1-D}
}
which diverges as $D$ approaches $0$. Therefore, even when the noise level (and $D$) is very small, it decreases the expected value significantly. On the other hand the first derivative of \eqref{eq:RenyiEntAsAMeaasureforGuess} is equal to zero at $p=1/2$ (i.e., when there is no bias). The first derivative around $p=1/2$ is very small and therefore bias does not affect the guesswork as much as noise. Figure \ref{fig:gwbiasvsnoise} presents the guesswork of an unstable unbiased PUF and the guesswork of a stable biased PUF. 

\begin{figure}
\bepsf{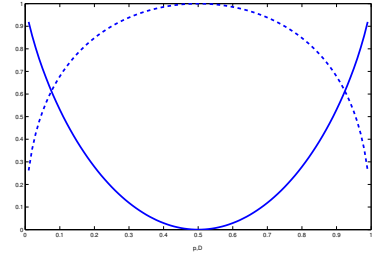}{8}\caption{The solid line presents the average guesswork of an unstable unbiased PUF $1-H\pa{D}$, whereas the dotted  line is the average guesswork of a stable biased PUF $H_{1/2}\pa{p}$.}\label{fig:gwbiasvsnoise}
\end{figure}

For example, the asymptotic exponential growth rate of the guesswork of an unbiased ($p=1/2$) unstable PUF with transition probability $D=0.1$ (i.e., a $10$\% noise) is equal to $0.53$ which is the guesswork of a stable biased PUF with $p=0.05$ (i.e., a $95$\% bias).   

In terms of min-entropy as presented in Lemma \ref{lem:ProofofMinEntropysubjecttoDistortion}, the divergence  $D\pa{D||p}=-H\pa{D}-D\log_{2}\pa{p}-\pa{1-D}\log_{2}\pa{1-p}$ and therefore its first derivative also diverge as $D$ goes to zero. Therefore, min-entropy is also very sensitive to the presence of noise. On the other hand, the min-entropy of a stable PUF is equal to
\beq{eq:noiselessmin_entropyExample}{
-m\log_{2}\pa{1-p}.
}
The first derivative of \eqref{eq:noiselessmin_entropyExample} equals $\frac{m}{1-p}$ when $0\le p\le 1/2$ and therefore it does not diverge. Hence, the effect of bias on the min-entropy is also less significant than the effect of noise. 

For example, the asymptotic min-entropy of an unbiased ($p=1/2$) unstable PUF with transition probability $D=0.1$ is equal to $1-H\pa{0.1}=0.53$ which is the min-entropy of a stable biased PUF with $p=0.31$ (i.e., a bias level of $69$\%).   

Note that in general the first derivative of the min-entropy does not equal to zero at $p=1/2$, and therefore bias has a stronger effect on min-entropy than on average guesswork. 

Figure \ref{fig:mebiasvsnoise} presents the behavior of the min-entropy as a function of $p$. It shows that it is more sensitive to bias than $H_{1/2}\pa{p}$.


\begin{figure}
\begin{center}
\epsfig{figure=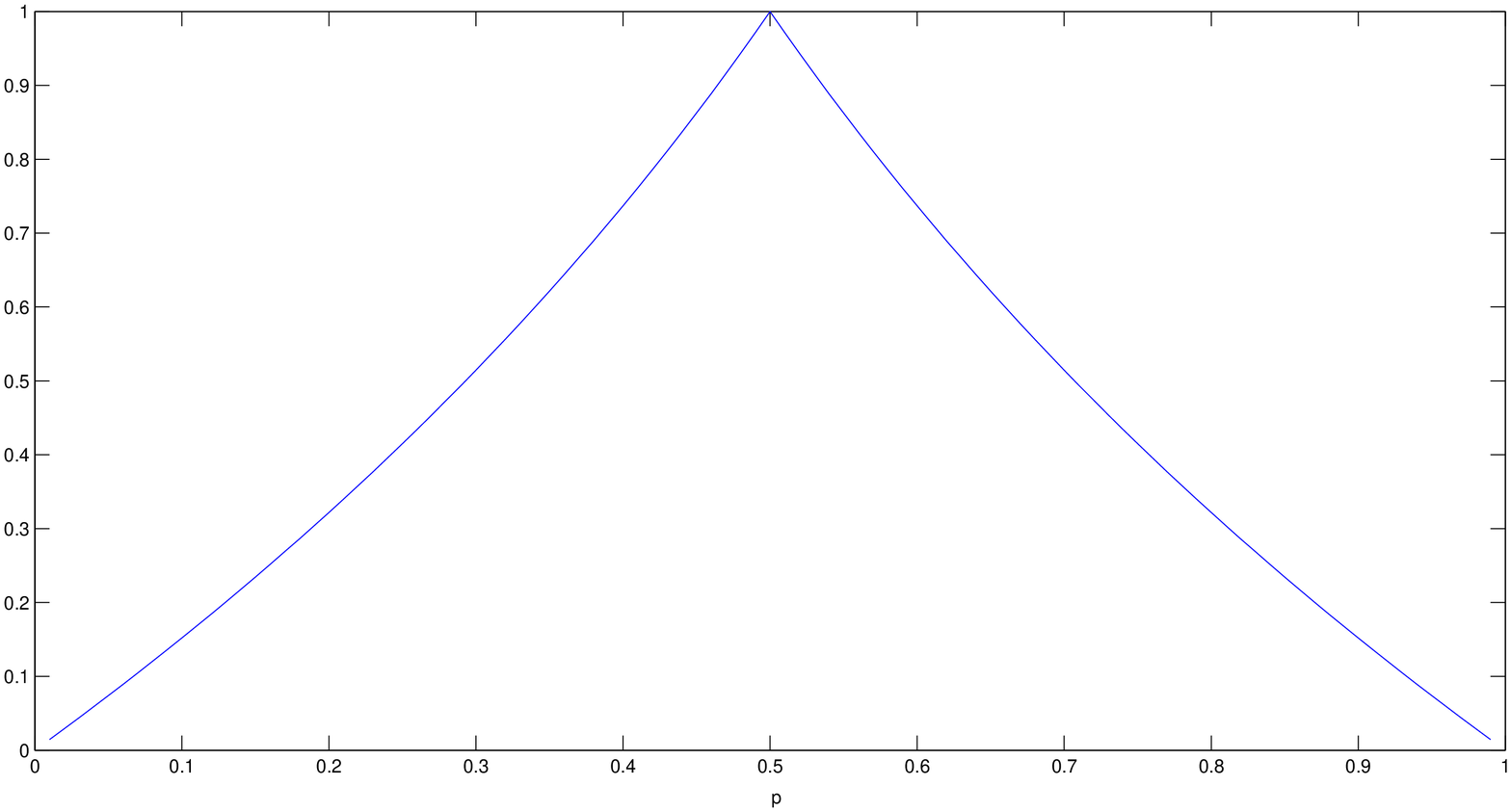,width=6 cm}
\end{center}
\label{fig:mebiasvsnoise}
\end{figure}

\section{Evaluating the Security Level of Weak PUFs Through Guesswork} \label{exp}

\subsection{Evaluation of Weak LEDPUF}
We evaluate the probability mass function of a bit generated by a weak LEDPUF based on simulation results for the formation of connections \beq{eq:bitsmarginals}
{
p_{X}\pa{1}=0.4626\quad p_{X}\pa{0}=0.5374.
}

The uniqueness is evaluated by calculating the fractional inter-distance \cite{Maiti2009} of 1000 weak LEDPUFs, each produces 512 bits of response. The distribution is with mean=0.503 and standard deviation=0.02. Since the variance value is proportional to the inverse of the length of the response, as the length of the response increases the variance value goes to zero while the mean value goes to $0.505$.



\subsection{Measurements of Noisy Weak PUFs}
Two noisy silicon PUFs: SRAM PUF and RO PUF, are measured at 20\textdegree C in our experiments. For the SRAM PUF, responses from 10 commercial 45nm SOI test chips with 176k byte of SRAM cells each are obtained. Every SRAM PUF is measured 10 times. The RO PUF is implemented on 15 Altera DE2-115 FPGA boards. To avoid correlated CRPs, 90 CRPs are generated from the 91 ROs in each RO PUF. Each RO PUF is measured 10 times.



The measurement results of noisy PUFs are summarized in Table \ref{table:inter_noisy}. The intra-FHD and inter-FHD are given in the second and third columns, respectively. Both PUFs show good results of small intra-FHD and close to 50\% inter-FHD. The stability shown in the forth column gives the percentage of stable bits through all 10 measurements, where a stable bit is a bit that remains the same during all measurements. A 93\% stability for the SRAM PUF, for example, means that 7\% of the bits flip at least once during the 10 measurements. For LEDPUF, the intra-FHD is 0\% and the stability is 100\%. The bias level (percentages of ones and zeros) are given in the last two columns.


\begin{table}
\centering
\caption{Noisy PUF measurements. All numbers are percentages.}
\begin{tabular}{|c|c|c|c|c|c|}
\hline
\multirow{2}{*}{} & \multirow{2}{*}{intra-FHD} & \multirow{2}{*}{inter-FHD} & \multirow{2}{*}{Stability} & \multicolumn{2}{c|}{Bias Level} \\ \cline{5-6} 
                  &         &           &           & One    & Zero \\ \hline
SRAM PUF          &  2.26 &  48.33  &  93.42  &  49.13  &   50.87   \\ \hline
RO PUF            &  2.48 &  47.13  & 91.19   &  51.38  &  48.62    \\ \hline
\end{tabular}
\label{table:inter_noisy}
\end{table}

For the RO PUF, in addition to the intra-FHD at 20\textdegree C, we also compare the intra-FHD between 20\textdegree C and 60\textdegree C, which is the reliability of the PUF if it is enrolled at 20\textdegree C but verified at 60\textdegree C. The results are presented in Figure \ref{fig:temp}. We can see that for most PUFs, the averaged intra-FHD at the extreme temperature is about 12\%, which implies that conventional ECC margin with error reduction techniques for the PUF would be required.

\begin{figure}[htb]
\centering
\includegraphics[width=2.5in]{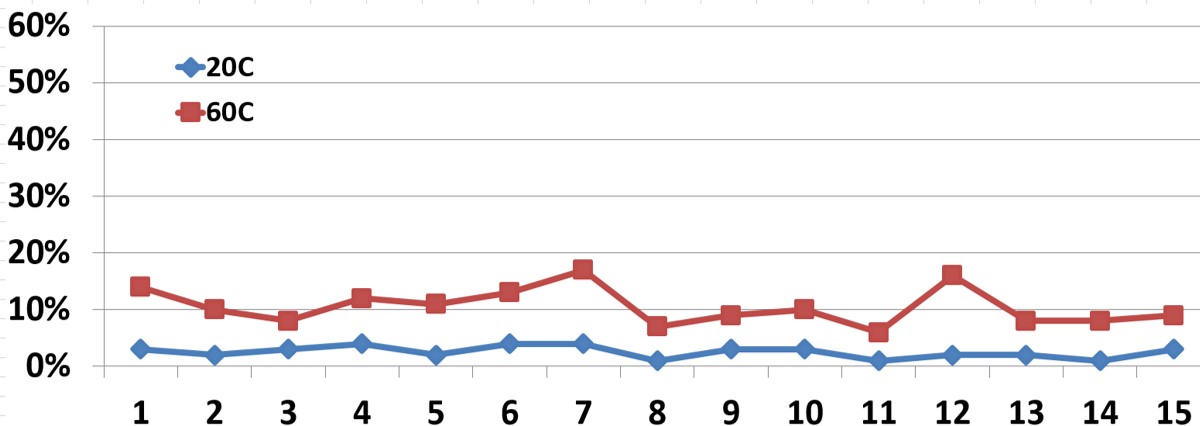}
\caption{Intra-FHD at extreme temperature variation for 15 RO PUFs.}
\label{fig:temp}
\end{figure}

For noisy SRAM PUF and RO PUF, the expected growth rate is calculated by plugging the intra-FHD to equation \eqref{eq:UpperBoundGuessWithDist}. The expected growth rate of weak LEDPUF is obtained by applying the bit probabilities given in \eqref{eq:bitsmarginals} to equation \eqref{eq:RenyiEntAsAMeaasureforGuess}. The results are summarized in Table \ref{table:guess}. We can see that even though the weak LEDPUF is more biased than the noisy PUFs, its guesswork growth rate is still higher than noisy PUFs. For RO PUF at 60\textdegree C, the guesswork growth rate becomes much worse compared with RO PUF at 20\textdegree C, which implies quantitatively how insecure a PUF can become under environmental variations.




\begin{table}
\centering
\caption{Growth rate of the expected value of the guesswork. When the key size of the PUF is 32, the average guesswork of the SRAM PUF is proportional to $2^{32\times0.8442}$, and the average guesswork of the LEDPUF is proportional to $2^{32\times0.9980}$.}
\begin{tabular}{|c|c|c|c|c|}
\hline
PUF Type & SRAM & RO at 20\textdegree C & RO at 60\textdegree C & LEDPUF\\
\hline
Growth rate & 0.8442  & 0.8323 & 0.4706 & 0.9980 \\
\hline
\end{tabular}
\label{table:guess}
\end{table}

\section{Strong LEDPUF Construction and Guesswork Analysis}\label{sec:GWStrongLEDPUF}
\subsection{The Construction of A Strong LEDPUF}
One of the shortcomings of using memory-based PUFs for CRPs, is the scaling of the hardware size as a function of the number of CRPs \cite{PUFStudy2010}. In general, each channel response pair requires a different set of circuits, and as a result the hardware size is proportional to the number of CRPs. On the other hand for strong PUFs the hardware size scales logarithmically as a function of the number of CRPs.

In order to create a strong LEDPUF we consider a keyed hash function along with a weak LEDPUF. The weak LEDPUF response is used as a key for the keyed hash function. The challenges serve as the input to the hash function, whereas the response is the output of the keyed hash function. Figure \ref{fig:Strong_LEDPUF_Scheme} presents a strong LEFPUF based on a keyed hash function and on a weak LEDPUF with $N$ bits of input and $L$ bits of output. 


It is important that the keyed hash function uses the key in such a way that does not enable the attacker to predict responses to unobserved challenges based on the observed ones. Therefore, concatenating the key directly to the challenge, which is vulnerable to extension or collision attacks, is not a good realization of the strong LEDPUF. 

We create strong LEDPUF by using a weak LEDPUF as a key for a keyed-hash message authentication code (HMAC) \cite{CanettiHMAC}. Any cryptographic hash function, such as SHA-1 or SHA-2 can be used for HMAC. It is worth mentioning that in \cite{JungHMAC} the authors also propose the use of PUFs with an HMAC in a somewhat similar manner; however, they do not take into consideration the overhead incurred by the instability of parametric PUFs.


To give a rough estimation of the hardware implementation cost of the strong LEDPUF, for a HMAC-SHA1, the implementation requires about 30k gates \cite{Mao-Yin2004}, and just the ECC part, BCH for example, of a parametric PUF would require same order of gates \cite{Xinmiao2015}.

The level of security of a strong LEDPUF depends on the underlying hash function and the quality of the weak LEDPUF that serves as a key, whereas weak LEDPUFs rely solely on the randomness in the manufacturing process.



\begin{figure}
\centering
\includegraphics[width = 2.8in]{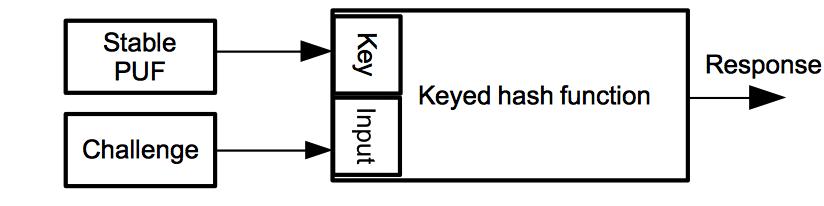}
\caption{A strong LEDPUF based on a keyed hash function (HMAC or NMAC) and a weak LEDPUF.}
\label{fig:Strong_LEDPUF_Scheme}
\end{figure}


For the simulation results, each strong LEDPUF consists of a weak LEDPUF that provides 2x256 bits for the two initial vectors (IV) of the nested hash, and each response is a 256-bit stream because SHA-256 is used in the construction. The same challenge is given to 1000 strong LEDPUFs, and the inter-distance of the responses is a distribution with mean=0.500 and standard deviation=0.03.

To construct a strong LEDPUF, only the weak LEDPUF can be used because of its 0\% intra-distance. If other existing weak PUFs with even small intra-distance are used, the intra-distance of the strong LEDPUF would be increased dramatically due to the avalanche effect. In other words, even a single bit flip of the weak PUF can completely change the response of the strong LEDPUF. Figure \ref{fig:weak_strong} (a) shows that the intra-distance of the strong LEDPUF jumps from 0\% to 50\% as the number of bit flips increases from zero to one.

Figure \ref{fig:weak_strong} (b) shows how the intra-distance of the strong LEDPUF rises as the intra-distance of the weak PUF increases in logarithmic scale. Since 2x256 bits of the IVs are from the weak PUF, for a weak PUF with 0.1\% intra-distance, the probability that it generates a same 512-bit response twice is about 60\%, which translates to a roughly 20\% intra-distance of the strong LEDPUF. Therefore, only the weak LEDPUF with a guaranteed 0\% intra-distance can be used for the IV generation.

\begin{figure}[h]
\begin{centering}
\includegraphics[width=3.5in]{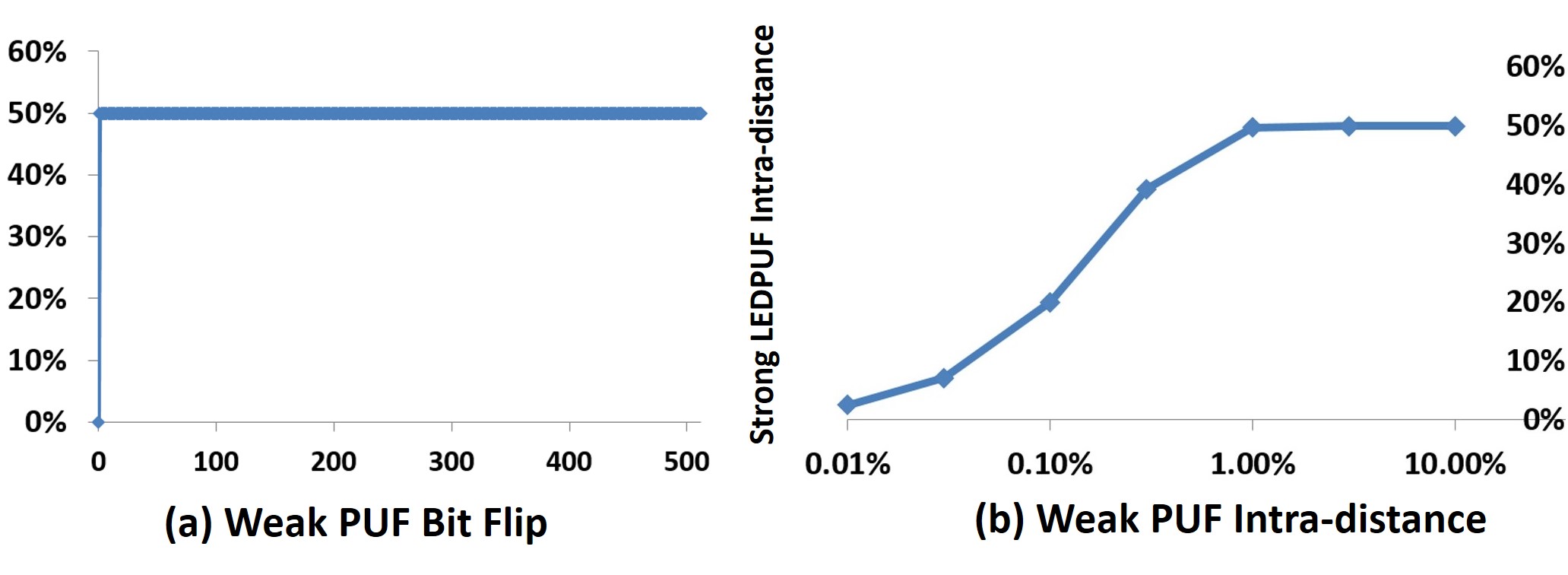}
\par\end{centering}
\caption{(a) A single bit flip from the weak PUF can induce a completely different response of the strong LEDPUF due to the avalanche effect of the hash function. (b) Intra-distance of the strong LEDPUF rises dramatically if other weak PUFs with small intra-distances are used in the strong LEDPUF construction.}
\label{fig:weak_strong}
\end{figure}
\vspace{-0.2in}
\subsection{The Guesswork of any Strong PUF}\label{subsec:TheGuessworkofAnyStrongPUF}
In this subsection we quantify the security of strong PUFs in terms of the number of guesses required to break them.
Our results quantify the number of secure authentications for which any strong PUF is good for. Furthermore, we compare the guesswork of our proposed strong LEDPUF to the guesswork of other strong PUFs that have been introduced in the literature. Finally, to demonstrate the importance of stability of a strong PUF, we show that the guesswork of a stable XOR arbiter PUF is larger than the guesswork of noisy ones, for the same number of observed CRPs.

We begin by defining the following game. 
\begin{defn}\label{def:AuthenticationGame}
Consider a strong PUF, which is used by an authentication scheme to authenticate $n$ unique challenges through observing their responses.  
The authentication problem is defined as follows:
\begin{enumerate}
\item For each challenge the attacker has to guess with a single response.
\item When the attacker does not guess correctly, it can mask itself to receive a new challenge.
\item Once the attacker makes a correct guess it is authenticated. 
\end{enumerate}
\end{defn}

\begin{rem}
Note that when authenticating a strong PUF through CRPs each challenge can be used only once. Furthermore, the problem defined above captures a strict security requirement that the system is compromised once the attacker manages to deceive the verifier. 
\end{rem}

\begin{rem}
When the attacker fails to guess any response correctly the attack fails. When the number of challenges is large we show that this event decays exponentially fast. 
\end{rem}

We now find the average guesswork of the game presented in Definition \ref{def:AuthenticationGame} as well the probability that the number of guesses is smaller than or equal to a certain number. 

\begin{theorem}
The average guesswork of the authentication problem presented in Definition \ref{def:AuthenticationGame} is
\beq{eq:AverageGuessworkGeneralExpressions}{
E\pa{G}=2^{-H_{\infty}\pa{1}}+\sum_{i=2}^{n}i\cdot 2^{-H_{\infty}\pa{i}}\cdot\prod_{k=1}^{i-1}\pa{1-2^{-H_{\infty}\pa{k}}}
}
where $2^{-H_{\infty}\pa{k}}$ is the most probable response to the $k$th challenge either given that the guesses for challenges $1,\dots, k-1$ were incorrect or given the previous $k-1$ CRPs (these two scenarios can lead to different min-entropy). Furthermore, the probability that the number of guesses is smaller than or equal to $l$ is 
\baln{}
{
Pr\pa{G\in\pac{1,\dots,l}}
&=&\nn\\
2^{-H_{\infty}\pa{1}}&+&\sum_{i=2}^{l}2^{-H_{\infty}\pa{i}}\prod_{j=1}^{i-1}\pa{1-2^{-H_{\infty}\pa{j}}}
}
Finally, the probability of attack failure is 
$\prod_{i=1}^{n}\pa{1-2^{-H_{\infty}\pa{i}}}$.
\end{theorem}

\begin{proof}
Each response has a certain statistical profile based on the previous CRPs. When the attacker knows this profile the optimal strategy to minimize the number of guesses is to guess the most probable one. This in turn leads to $2^{-H_{\infty}\pa{i}}$ where $H_{\infty}\pa{i}$ is the min-entropy given that either the previous $i-1$ guesses were not correct or that the previous $i-1$ CRPs were revealed to the attacker. The results of the theorem follow directly from this argument. 
\end{proof}

\begin{cor}
When the statistical profile does not change across challenges we get that
\beq{eq:AverageGuessworkNonUniform}{
E\pa{G} = 2^{H_{\infty}}-\pa{1-2^{-H_{\infty}}}^{n}\cdot\pa{n+2^{H_{\infty}}}
}
and 
\beq{eq:ProbabilityGuessworkSmallerThanCertainValue}
{
Pr\pa{G\in\pac{1,\dots,l}}=1-\pa{1-2^{-H_{\infty}}}^{l}.}
\end{cor}
\begin{rem}
Note that the average guesswork \eqref{eq:AverageGuessworkNonUniform} is equal to 
\beq{eq:AverageGuessworkConvergeToMinEntropyiid}{
E\pa{G}= 2^{H_{\infty}}-\epsilon
}
when $n\gg 1$, where $\epsilon\ll 1$ decays exponentially fast.
\end{rem}

\begin{rem}
When the attacker does not know the statistical profile of the response (for instance when the structure is too complex for him to infer it), all he can do is guess a response uniformly, which leads in turn to $H_{\infty}=1$.

Furthermore, in some cases the attacker can infer the statistical profile based on the structure of a strong PUF and a set of CRPs that have been revealed to him (see \cite{Ruhrmair13,Becker2015} for attacks on arbiter PUFs, XOR arbiter PUFs, etc.).  
\end{rem}

\subsection{Quantifying the Security of Specific Strong PUFs}
In this subsection we quantify the level of security of various strong PUFs in terms of their guesswork. 

the result in  \eqref{eq:AverageGuessworkConvergeToMinEntropyiid} also applies to the case when an attacker can observe multiple CRPs. For example model attacks over strong PUFs \cite{Ruhrmair13,Becker2015} enable attackers to accurately guess responses based on previously observed CRPs. In terms of guesswork it means that once an attacker observes a certain set of CRPs, conditioned on the CRPs that have been revealed so far, the most likely conditional probability can be very high. 

The average guesswork in \eqref{eq:AverageGuessworkConvergeToMinEntropyiid} allows us to quantify how secure strong LEDPUF is compared to other strong PUFs from the literature that are susceptible to model building attacks. 

Strong LEDPUFs are based on HMAC, and so idealy an attacker can not infer anything from observing CRPs in this case. Therefore, when the number of bits at the output of a strong LEDPUF is $m$, the average guesswork  converges to $2^{m}$ after observing any amount of CRPs. Even when the key is biased such that each bit is drawn Bernoulli$\pa{0.53}$ \cite{Wang16}, the average guesswork when HMAC is a strongly universal set of hash functions \cite{Carter_Univ_Hash_Func}, is $2^{-\log_{2}\pa{0.53}\cdot m}=2^{0.91\cdot m}$ based on \eqref{eq:AverageGuessworkConvergeToMinEntropyiid}.

On the other hand, in \cite{Ruhrmair13} it has been shown for various noise-free PUF simulations such as arbiter PUFs, XOR arbiter PUFs and Feed-Forward Arbiter PUFs that the prediction rate varies between $97\%$ and $99\%$, after observing a few hundreds of thousands of CRPs and implementing a model building attack. Essentially, this type of attacks achieve a prediction rate, which is an estimation of the probability mass function of the next response, conditioned on a certain set of CRPs. In \cite{Ruhrmair13} it leads to conditional probability of at least $2^{\log_{2}\pa{0.97}\cdot m}=2^{-0.04\cdot m}$ when the prediction rate is $97$\%, and so the average guesswork under this model building attack is achieved by assigning this probability to the average guesswork in \eqref{eq:AverageGuessworkConvergeToMinEntropyiid} in which case we get $2^{0.04\cdot m}$. Note that when the prediction rate is $99$\% the average guesswork is $2^{-\log_{2}\pa{0.99}\cdot m}=2^{0.014\cdot m}$; therefore, when $m=256$ we get that at $97$\% the PUF is $2^{0.026\times 256}=100$ times more secure than at $99$\%.

Therefore, guesswork provides a unified framework for comparing the level of security of different PUFs under model building attacks. It can also be used as a means of understanding what is the desired prediction rate for a model based attack, and as a result how many challenge response pairs should be observed.

Next, we compare the security of stable and noisy XOR arbiter PUFs \cite{Ruhrmair13} under model based attacks in terms of the number of guesses for which the probability of guessing the correct response is $99$\%. We use equation \eqref{eq:ProbabilityGuessworkSmallerThanCertainValue} to derive the results of this subsection. 

The expression in \eqref{eq:ProbabilityGuessworkSmallerThanCertainValue} depends on the min-entropy, and so for noisy PUFs we need to incorporate the effect of the noise into the min-entropy. For this we use Lemma \ref{lem:ProofofMinEntropysubjecttoDistortion} in which the min-entropy is extended to the noisy case. In Table \ref{table:guess} we use guesswork to compare the security of stable XOR arbiter PUFs to the one of noisy XOR arbiter PUFs with the same number of XORs, under model based attacks. We assign the prediction rates and noise levels reported in \cite{Ruhrmair13} to \eqref{eq:ProbabilityGuessworkSmallerThanCertainValue}, and find the number of guesses under model based attacks in which the verifier has to take into account the noise as the PUF owner observes noisy responses. The table shows how much more secure stable XOR arbiter PUFs are when compared to the noisy versions under model based attacks. Essentially, it shows how susceptible such arbiter-based strong PUF is after observing a certain number of CRPs. In fact, when the noise level is over 5\%, the probability of guessing the correct response up to the noise level is very close to one (about $1-10^{-10}$), which means that this PUF is completely broken. This is because the prediction rate of each bit as reported in \cite{Ruhrmair13} is $97.34$\%, whereas the noise level is $5$\% and so the chance that the guessed response is not within the noise level of the original response is extremely small for reasonable values of $m$. Therefore, guesswork enables us to incorporate the effect of noise and model based attacks into one framework that allows comparison between different PUFs and various scenarios. 

\begin{table}
\centering
\caption{The number of guesses for which the probability of guessing the correct response is lager than $99$\% when $m=1024$, for a stable XOR arbiter PUF ($D=$ 0\%), and noisy ones, under model based attacks for which $200$ thousand and $500$ thousand CRPs are observed. The values are based on the noise levels and prediction rates reported in \cite{Ruhrmair13}. The second row presents 4-XORs, whereas the third row presents 5-XORs.}
\begin{tabular}{|c|c|c|c|c|}
\hline
CRPs ($\times 10^{3}$) & D=0\% & D=2\% & D=5\% & D=10\%\\
\hline
200 & $41$  & $10$ & $1$ & $1$ \\
\hline
500 & $22$ & $10$ & $1$ & $1$\\
\hline
\end{tabular}
\label{table:guess}
\end{table}

\subsection{Guesswork of Message Authentication Codes}
When the secret of a strong PUF can be read once in a secured and tamper-evident manner, it can serve as a message authentication code. When a keyed hash function is used as a message authentication code (MAC), ideally, a key chooses a hash function from a strongly universal set of hash functions \cite{Carter_Univ_Hash_Func}, that is, the set of all possible hash functions with an input of size $N$ and an output of size $L$, where $N
\ge L$. In this subsection we analyze the guesswork of MACs under the theoretical assumption that the key
chooses a hash function from this set based on a mapping, $f\pa{\cdot}$, from the set of all possible keys to the strongly universal set of hash functions. 

We begin by defining a guessing game for strong LEFPUFs.
\begin{defn}
Consider the following game: Bob draws a key at random from a probability function $P_{X}\pa{x}$, and an attacker
Alice who does not know $x$ but knows the probability mass function $P_{X}\pa{\cdot}$ and the mapping from the set of keys to the set of hash functions, $f\pa{\cdot}$, tries to guess the responses of a dictionary of size $2^{N}$ (i.e., to guess the responses of $2^{N}$ messages that may be sent). For each value among the set of $2^{N}$ possible inputs, an
oracle tells Alice whether her guess is correct. When Alice guesses the right output, Bob proceeds to the next message. The game ends when Alice guesses the response for each and every message. We assume that Alice has an unlimited computing power and so she knows the mapping from the set of keys to the set of bins (i.e., for each challenge she knows the underlying distribution of the bins); however, she does not know the actual value of the key.
\end{defn}
\begin{rem}
In this problem Alice tries to find the responses to a set of messages that may be used by a user. After she breaks these messages, she can impersonate as the user. Since these messages may be used multiple times (e.g., popular words), she can try to guess the response to each of the messages multiple times as well. This is different from the problem presented in subsection \ref{subsec:TheGuessworkofAnyStrongPUF}, where the attacker had a single guess for each challenge. 
\end{rem}

The optimal strategy for minimizing the number of guesses is minimizing the number of guesses for each message. This is achieved by guessing for each message the bins starting from the most probable bin to the least probable one.

The guesswork is the sum of the number of guesses across all messages, and therefore is equal to \vspace{-0.1in}
\beq{eq:TheTotalGuesswork}
{
G\pa{bins}=G\pa{bin_{1}}+\sum_{k=2}^{2^{N}}G\pa{bin_{k}|bin_{1},\dots,bin_{k-1}}.
}
where $G\pa{bins}$ is the number of guesses required to find the responses of all $2^{N}$ messages, $G\pa{bin_{1}}$ is the number of guesses for the first message, and $G\pa{bin_{k}|bin_{1},\dots,bin_{k-1}}$ is the number of guesses for the $k$th message given the bin values for the first $k-1$ messages. Note that 
\beq{}
{
P_{BINS}\pa{bin_{1},\dots,bin_{2^{N}}}=P_{BINS}\pa{f\pa{x}}=P_{X}\pa{x}
}
where $f\pa{\cdot}$ is the mapping from keys to bins, which is a bijection. Hence, the vector of bins and the key have the same probability mass function.

By averaging we get that
\baln{}
{
\pa{1+L\ln 2}^{-1}&\cdot& \left( 2^{H_{1/2}\pa{bin_{1}}}+\sum_{k=2}^{2^{N}}2^{H_{1/2}\pa{bin_{k}|bin_{k-1},\dots,bin_{1}}}\right)  \nonumber\\
\le &E&\pa{G\pa{bins}} \le \nonumber\\
2^{H_{1/2}\pa{bin_{1}}}&+&\sum_{k=2}^{2^{N}}2^{H_{1/2}\pa{bin_{k}|bin_{k-1},\dots,bin_{1}}}.
}

Without loss of generality let us write the conditional sum over the messages in \eqref{eq:TheTotalGuesswork} as
\beq{eq:TotalGuessworkSimpleWriting}
{
G\pa{bins}=\sum_{k=1}^{2^{N}}bin_{k}.
}
Based on  \eqref{eq:TotalGuessworkSimpleWriting} we can state the following two theorems. 
\begin{theorem}\label{th:MartUniformCase}
When the key is i.i.d. Bernoulli$\pa{1/2}$, the average guesswork across all messages for any $f\pa{\cdot}$, which is a bijection, is
\beq{eq:AverageUnbiasedKey}{
\eta=E\pa{G\pa{bins}}=2^{N}\cdot\pa{\frac{2^{L}+1}{2}}
}
and the probability of deviating from the mean value is upper bounded by
\beq{eq:AzimasIneqUnBiasedKeys}{
Pr\pa{\pa{G\pa{bins}-\eta}>\alpha\cdot\eta}\le\exp\pa{-\frac{\alpha^{2}}{8}2^{N}}
}
where $\Al\ge 0$.
\end{theorem}
\begin{proof}
When the key is i.i.d. Bernoulli$\pa{1/2}$, and $f\pa{\cdot}$ is a bijection, the bins $\pac{bin_{1},\dots,bin_{2^{N}}}$ are i.i.d. uniform. Therefore, the average guesswork of each bin is $\frac{2^{L}+1}{2}$, and since there are $2^{N}$ messages we get \eqref{eq:AverageUnbiasedKey}. Furthermore, the probability function of $G\pa{bins}$ can be upper bounded by  using Azuma's inequality \cite{AlonProbMethod}. The sum of the elements in \eqref{eq:TotalGuessworkSimpleWriting} is a Martingale (when subtracting the expected value) and therefore the probability function can be bounded by Azuma's inequality
\beq{}
{
Pr\pa{\pa{G\pa{bins}-\eta}>t}\le \exp\pa{-\frac{t^{2}}{8\cdot 2^{N+2L}}}.
}
Therefore, for $t=\alpha\cdot \eta$ and $\alpha>0$ we get \eqref{eq:AzimasIneqUnBiasedKeys}.
\end{proof}

Equation \eqref{eq:AzimasIneqUnBiasedKeys} leads to a concentration result, showing that as $N$ increases, the probability of deviating from the average guesswork goes to zero. 

\begin{theorem}
When the key is drawn i.i.d. Bernoulli$\pa{p}$, where $0< p \le 1/2$, and $f\pa{\cdot}$ is the identity function, that is, $\pa{bin_{1},\dots,bin_{2^{N}}}=x$, the average guesswork increases like  
\beq{eq:AverageGuessworkiidey}
{
E\pa{G\pa{bins}}=\eta=2^{\pa{N+L\cdot H_{1/2}\pa{p}}}
}
and 
\beq{eq:AzimasIneqBiasedKeys}
{
Pr\pa{\pa{G\pa{bins}-\eta}>\alpha\cdot\eta}\le\exp\pa{-\frac{\alpha^{2}}{8}2^{N-2L\cdot\pa{1-H_{1/2}\pa{p}}}}
}
where $\Al >0$. 
\end{theorem}
\begin{proof}
When $f\pa{\cdot}$ is the identity mapping, $bin_{1},\dots,bin_{2^{N}}$ are i.i.d. (regardless of the order of messages). The rest of the proof follows along the same lines as the proof of Theorem \ref{th:MartUniformCase}. 
\end{proof}

\vspace{-0.2in}
\section{Conclusion} \label{conclusion}
In this paper we propose the first stability-guaranteed PUF that requires no stability enhancement techniques, where the source of randomness is extracted from locally enhanced DSA process. 
Detailed constructions of two LEDPUFs: the weak LEDPUF and the strong LEDPUF, are presented. Inter-distance measurements on the LEDPUFs show that both weak and strong LEDPUFs are ideally unique. The area and latency of the weak LEDPUF is much smaller than existing weak PUFs because no error correcting schemes are needed. The strong LEDPUF provides large CRP space because of its cryptographic hash based structure. The weak LEDPUF used in the strong LEDPUF construction cannot be replaced by existing weak PUFs because an absolute 0\% intra-distance is required for the weak PUF to avoid the avalanche effect of the strong LEDPUF. We quantify the level of security provided by weak LEDPUF by calculating the expected guesswork resulting from weak LEDPUFs empirical probability function; the loss compared to a fair coin toss is negligible.

Furthermore, we develop a unified guesswork-based analyses for PUFs. We show through guesswork analysis that stability has a more severe impact on the PUF security than biased responses. In addition, we analyze guesswork for three new problems: Guesswork under probability of attack failure, the guesswork of strong PUFs, and the guesswork of idealized version of MACs.

Our ongoing work includes exploring other sources of LEDPUF that are more resilient to invasive attacks, such as Scan Electron Microscopy (SEM) or Transmission Electron Microscopy (TEM). 
\vspace{-0.2in}
\begin{appendices}
\section{Proof of Theorem \ref{th:GWWithErrors}}\label{sec:ProofThGWWithErrors}
The proof relies on the method of types \cite{CoverBook} as well as on the following assumptions. 
\begin{enumerate}
\item When $\Al=D\pa{s||p}$ and $s> p$ the types for which $D\pa{q||p}\ge D\pa{s||p}$ are $p <  s\le q\le 1$ as well as any $0\le q <p$ for which $D\pa{q||p}\ge D\pa{s||p}$. 
\item Since the probability that a type is in $A$ decays like $2^{-m\cdot D\pa{s||p}}$, the conditional probability of drawing a certain vector of type $q$ given that $q$ is not in $A$ decreases like $2^{-\pa{H\pa{q}+D\pa{q||p}}\cdot m}$ \cite{CoverBook} (i.e., it decays at the same rate as the original probability mass function). 
\item When there are no constraints, the rate at which the average guesswork of any moment increases, is the solution to the following optimization problem
\beqn{}{
\max_{0\le q\le 1}\rho\cdot H\pa{q} -D\pa{q||p}=\rho\cdot H_{1/\pa{1+\rho}}\pa{p}
}
where $\rho\cdot H\pa{q} -D\pa{q||p}$ is a concave function whose maximum is at $s^{\ast}$. 
\end{enumerate}

First let us consider the case where $D=0$.
Since the rate at which the conditional probability mass  function decreases does not change (as stated above), the average guesswork of any moment is the solution to the following optimization problem.
\beq{}{
\max_{q\not\in A}\rho\cdot H\pa{q}-D\pa{q||p}. 
}
When $s^{\ast}\not\in A$ (i.e., when the attacker can guess $s^{\ast}$) which occurs in the range $s{\ast}\le s\le 1$, the average guesswork does not change although the probability of attack failure is no longer zero. However, when $p<s<s^{\ast}$, $s^{\ast}\in A$ and since the function is concave the solution to the optimization problem is $s$, which is the closest element to $s^{\ast}$ among the elements that are not in $A$.

In the case when $D>0$ the proof follows along the same lines as Theorem 1 in \cite{MerhavArikanDist}. This is because the conditional probability remains the same and therefore for the binary case with Hamming distance, the problem can be reduced to the following optimization problem
\beq{}{
\max_{q\not\in A}\pa{\rho\cdot H\pa{q}-D\pa{q||p}}-\rho\cdot H\pa{D}.
}

\vspace{-0.2in}
\section{Proof of Lemma \ref{lem:ProofofMinEntropysubjecttoDistortion}}\label{app:ProofofMinEntropysubjecttoDistortion}
The proof is based on the method of types \cite{CoverBook}.
We need to find the exponential rate at which the sum
\beq{eq:ballprobsum}{
\sum_{i=0}^{m\cdot D}\binom{m}{i} p^{i}\pa{1-p}^{m-i}
}
decreases. Let us define $i=\alpha\cdot m$. From the method of types we know that the exponential growth rate of $\binom{m}{\alpha\cdot m}$ is $H\pa{\alpha}$. On the other hand the rate at which $ p^{\alpha m}\pa{1-p}^{\pa{1-\alpha}m}$ decreases is $H\pa{\alpha}+D\pa{\alpha||p}$. We are interested in the most dominant term in \eqref{eq:ballprobsum}, and therefore we wish to find
\beq{}{
\min_{0\le\alpha\le D}D\pa{\alpha||p}.
}
The solution to this optimization problem is the statement of 
the lemma. 
\end{appendices}

\vspace{-0.2in}
\bibliographystyle{unsrt}
\small
\bibliography{PUF}

\end{document}